\documentclass[11pt]{article}
\usepackage{amsfonts}
\usepackage{amsmath}
\usepackage{enumerate}
\usepackage{mathabx}
\usepackage{listings}
\usepackage{graphicx}
\usepackage{color}

\usepackage{natbib}

\usepackage{bm}
\usepackage{verbatim}
\usepackage{enumitem}

\newcommand{\handout}[5]{
   \renewcommand{\thepage}{#1-\arabic{page}}
   \noindent
   \begin{center}
   \framebox{
      \vbox{
    \hbox to 5.78in { {\bf 6.893 Sub-linear Algorithms}
     	 \hfill #2 }
       \vspace{4mm}
       \hbox to 5.78in { {\Large \hfill #5  \hfill} }
       \vspace{2mm}
       \hbox to 5.78in { {\it #3 \hfill #4} }
      }
   }
   \end{center}
   \vspace*{4mm}
}

\newtheorem{theorem}{Theorem}

\newtheorem{remark1}[theorem]{Remark}
\newtheorem{example}[theorem]{Example}

\newcommand{\qed}{\rule{7pt}{7pt}}

\newenvironment{proof}{\noindent{\bf Proof}\hspace*{1em}}{\qed\bigskip}
\newenvironment{proof-sketch}{\noindent{\bf Sketch of Proof}\hspace*{1em}}{\qed\bigskip}
\newenvironment{proof-idea}{\noindent{\bf Proof Idea}\hspace*{1em}}{\qed\bigskip}
\newenvironment{proof-of-lemma}[1]{\noindent{\bf Proof of Lemma #1}\hspace*{1em}}{\qed\bigskip}

\newenvironment{proof-of-theorem}[1]{\noindent{\bf Proof of Theorem #1}\hspace*{1em}}{\qed\bigskip}

\newenvironment{proof-attempt}{\noindent{\bf Proof Attempt}\hspace*{1em}}{\qed\bigskip}

\makeatletter
\def\fnum@figure{{\bf Figure \thefigure}}
\def\fnum@table{{\bf Table \thetable}}
\long\def\@mycaption#1[#2]#3{\addcontentsline{\csname
  ext@#1\endcsname}{#1}{\protect\numberline{\csname
  the#1\endcsname}{\ignorespaces #2}}\par
  \begingroup
    \@parboxrestore
    \small
    \@makecaption{\csname fnum@#1\endcsname}{\ignorespaces #3}\par
  \endgroup}
\def\mycaption{\refstepcounter\@captype \@dblarg{\@mycaption\@captype}}
\makeatother

\newcommand{\mathify}[1]{\ifmmode{#1}\else\mbox{$#1$}\fi}
\newcommand{\bigO}O

\usepackage{hyperref}
\hypersetup{colorlinks, 
			citecolor=blue, 
			linkcolor=blue, 
			urlcolor=black}

\usepackage[titletoc,page]{appendix}

\RequirePackage[OT1]{fontenc}

\usepackage[space]{grffile}
\usepackage{subfig}
\usepackage[utf8]{inputenc}
\usepackage{authblk}

\title{Non-reversible, tuning- and rejection-free Markov chain  Monte Carlo via iterated random functions}
\author{Jelena Markovic\footnote{To whom the correspondence should be addressed. Email: jelenam@stanford.edu.} }
\author{Amir Sepehri}
\affil{\small Department of Statistics \\ Stanford University \\ 390 Serra Mall\\ Stanford CA 94305 }
\date{}

\providecommand{\keywords}[1]{\textbf{\textit{Keywords---}} #1}

\begin{document}

\maketitle

\begin{abstract}
In this work we present a non-reversible, tuning- and rejection-free Markov chain Monte Carlo (MCMC) which naturally fits in the framework of hit-and-run. The sampler only requires access to the gradient of the log-density function, hence the normalizing constant is not needed.
	 
We prove the proposed Markov chain is invariant for the target distribution and illustrate its applicability through a wide range of examples. We show that the sampler introduced in the present paper is intimately related to the continuous sampler of \cite{peters2012rejection, bouchard2017bouncy}. In particular, the computation is quite similar in the sense that both are centered around simulating an inhomogenuous Poisson process. The computation can be simplified when the gradient of the log-density admits a computationally efficient directional decomposition into a sum of two monotone functions. We apply our sampler in selective inference, gaining significant improvement over the formerly used sampler \citep{selective_sampler}. 
\end{abstract}



\section{Introduction}

Traditional Markov chain Monte Carlo (MCMC) methods such as Metropolis-Hastings are reversible by construction. Their ``diffusive'' behavior is known to slow down the convergence rate of the distribution of the samples of the chain to the target distribution. Non-reversible Markov chains gained much interest as it became known that they can mix much faster than the reversible ones \citep{diaconis2000analysis}. 

The present paper introduces a non-reversible, tuning- and rejection-free Markov chain. This sampler naturally combines the frameworks of iterated random functions \citep{diaconis1999iterated} and the hit-and-run \citep{andersen2007hit}. It is devised for sampling from a distribution by only using the gradient of its log-density, hence the normalizing constant of the density is not needed.

Markov chains can be constructed by iterating random functions on the state space $S$. For a family of functions $\{f_\theta\mid f_\theta: S \mapsto S , \theta \in \Theta \}$ and a distribution $\mu$ on the parameter space $\Theta$, one can construct a Markov chain as follows. From a current state $x \in S$ move to the next state $y = f_{\theta_0}(x)$, where $\theta_0$ is a random element of $\Theta$ distributed according $\mu$. For our purposes $\mu$ does not depend on $x$. This language provides a unifying tool for studying the properties of Markov chains. For example, \cite{diaconis1999iterated} show that under some conditions the constructed Markov chain will have a unique stationary distribution.

We introduce the proposed sampler via iterated random functions. Consider sampling from a density $\pi$ on $\mathbb{R}^d$. The new sampler can be constructed as follows. There is a family $\{f_{V,v}:(V,v)\in (0,1)\times\mathbb{S}^{d-1}\}$, where $\mathbb{S}^{d-1}$ is a unit sphere in $\mathbb{R}^d$, used together with $\mu$ being the uniform distribution on $\Theta = (0,1)\times \mathbb{S}^{d-1}$ to generate a Markov chain with $\pi$ as the stationary distribution. More specifically, starting from an initial state $X_0=x_0$, $X_1=f_{V_0,v_0}(x_0)$, $X_2=f_{V_1,v_1}(x_1)$ etc. Inductively, at each update step $n$, we sample the pair $(V_n,v_n)$ from $\mu$ and define
\begin{equation*}
	x_{n+1}=f_{V_n,v_n}(x_n).
\end{equation*}
The function $f_{V_n,v_n}$ is given as
\begin{equation} \label{eq:irf}
	f_{V_n,v_n}(x_n) = x_n + \tau\cdot v_n/2,
\end{equation}
where $\tau$ is a real number depending on $V_n$, $v_n$, and $x_n$; an explicit form for $\tau$ is given in Section \ref{sec:sampler}. As an example, when the target distribution is standard univariate Gaussian $\mathcal{N}(0,1)$ and the current state is $x_n= -1$, the function $f_{V_n,v_n}(x_n)$ is plotted as a function of $V_n$ and $v_n$ in Figure \ref{fig:iterated:random:function}.

\begin{figure}[h]
\centering
\includegraphics[width=0.6\textwidth]{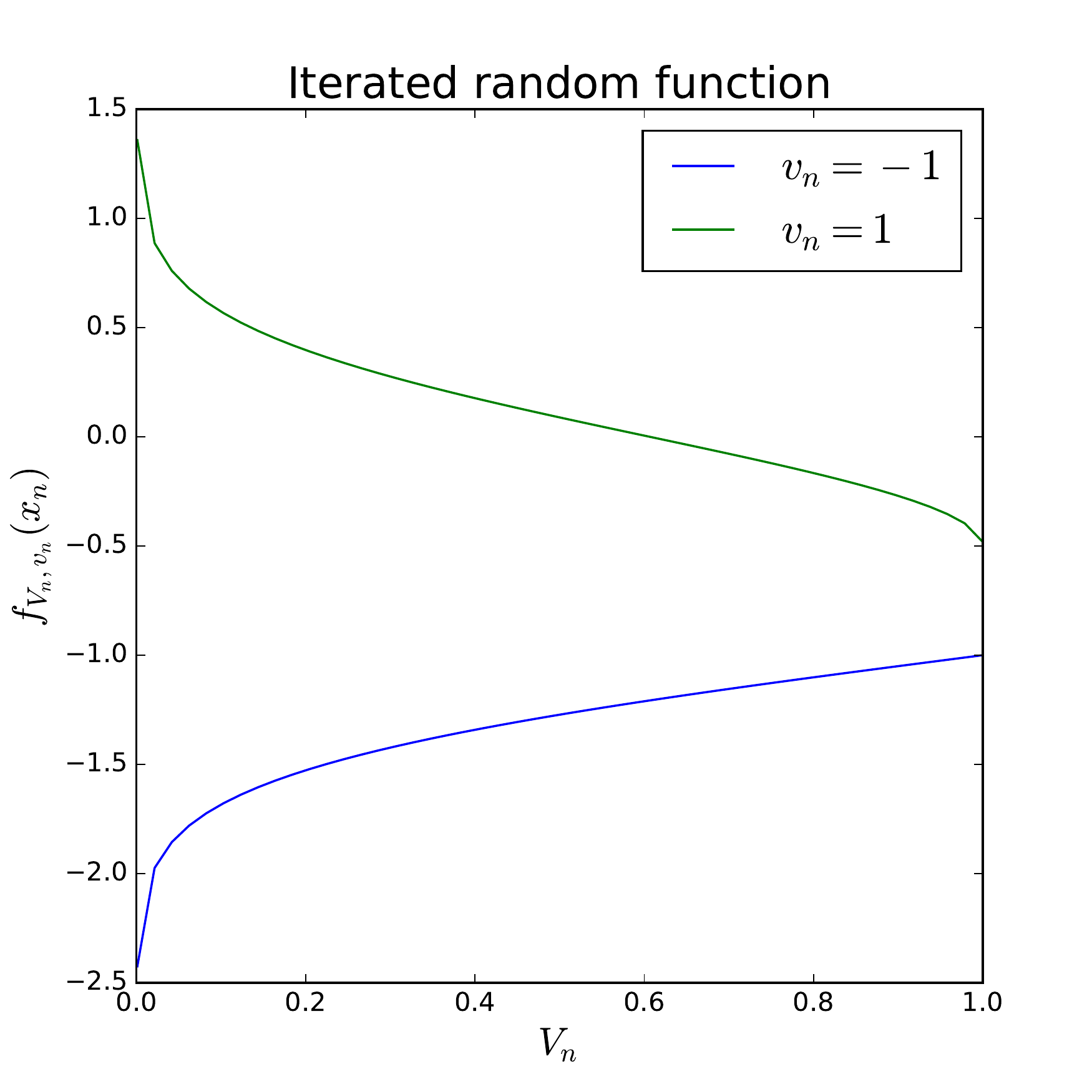}
\caption{Transition function $f_{V_n, v_n}(x_n)$ as a function of the variables $V_n$ and $v_n$ while $x_n$ is fixed at -1 for the chain corresponding to the standard Gaussian distribution $\mathcal{N}(0,1)$.}
\label{fig:iterated:random:function}
\end{figure}

This Markov chain has a natural description as a hit-and-run chain. The hit-and-run, in its simplest form, works as follows. Starting from a state $x$, pick a uniformly random direction $v$ and let $\pi_v$ be the restriction of $\pi$ to the line passing through $x$ in the direction of $v$. The next state is given by a sample from $\pi_v$. In many applications exact sampling from $\pi_v$ is not feasible, but hit-and-run only requires the next state to be generated from a Markov chain with $\pi_v$ as the stationary distribution. In our example, the parameter $(V,v)$ defines a line passing through the current state $x_n$, and the next state lies on that line. In particular, the new sampler in multiple dimensions is hit-and-run equipped with the same new sampler on each line.

Even though coming up with the non-reversible Markov chain is not straightforward, there is already a vast literature on various non-reversible Markov chain including theoretical guarantees about their favorable mixing properties \citep{duane1987hybrid, hwang1993accelerating, lelievre2013optimal}.
Since reversible Markov chains are easier to construct, i.e.~by finding a kernel that satisfies detailed balance equation, in some cases the researchers construct the non-reversible chains by extending the existing reversible ones \citep{chen2013accelerating, bierkens2016non, turitsyn2011irreversible, sun2010improving}.

The proposed sampler in this paper is related to a new class of non-reversible continuous-time Markov chains introduced in \cite{peters2012rejection} which represented a building block for the works of \cite{bouchard2017bouncy, bierkens2016zig, pakman2016stochastic, pakman2017binary, bierkens2017piecewise, wu2017generalized}. The chain here, as well as the Markov chains in their work, requires sampling from the distribution of the first arrival time of the Poisson process (PP) that depends on the log-density of the target distribution. All these works construct a continuous-time Markov chain so naturally they are not tuning-free since we might need to at least specify the discrete time points at which we evaluate the positions of the Markov chain path (Remark \ref{remark:discrete:vs:continuous:time}).

\subsection{Outline of the paper}

We introduce our main sampler in Section \ref{sec:sampler}, proving the target density is invariant with respect to its kernel in Section \ref{sec:invariance}.
We present a modified sampler with significant computational benefits that uses the decomposition of the negative log-density into a sum of two monotone functions in Section \ref{sec:sampling:via:decomposition}. The application of our sampler in selective inference is in Section \ref{sec:application}, where we sample from truncated log-concave density.

\section{Sampler}  \label{sec:sampler}

We describe precisely the iterated random function $f_{V_n,v_n}$ from \eqref{eq:irf}, driving the proposed sampler. Before describing $\tau$ as a function of the parameters $(V_n,v_n)\sim\mu$ and the current position $x_n$, we introduce some notation.
We write $\pi(x)$, $x\in\mathbb{R}^d$, as proportional to $e^{-U(x)}$, for a continuously differentiable ($C^1$) function $U(x):\mathbb{R}^d\rightarrow\mathbb{R}$, also referred to as a potential function. For now let us assume $U$ is $C^1$ on the whole of $\mathbb{R}^d$ but it suffices for $U$ to be $C^1$ on its domain. We discuss the truncated distributions in Section \ref{sec:truncated:log:concave}. We assume we have access to $U$ and its gradient $\nabla U(x)$.

Given $V_n, v_n$ and $x_n$, a positive real number $\tau$ becomes the following
\begin{equation} \label{eq:tau:equation}
	\tau=\min\left\{t\geq 0: -\log V_n \leq  \int_0^t\lambda(s)ds \right\}
\end{equation}
for $\lambda(s) = \left(\nabla U(x_n+v_n\cdot s)^Tv_n\right)_+$, where $F_+=\max\{F,0\}$ denotes the positive part of a function $F$.
Given $\tau$, the next point in the Markov chain becomes $x_{n+1}=x_n+ v_n\cdot\tau/2$. The interpretation is straightforward: given a random direction $v_n$, the Markov chain moves for time $\tau/2$ along that direction with speed $v_n$, moving for the total length of $v_n\cdot\tau/2$. The constructed Markov chain keeps the target distribution $\pi$ invariant (proof in Section \ref{sec:invariance}).

Before describing further interpretations of $\tau$, let us mention what the proposed chain looks like in two examples.

\begin{example}[Uniform distribution] We describe the update steps of the proposed Markov chain in the case of a univariate uniform distribution from $a$ to $b$ with probability density $\pi(x)=\frac{1}{b-a}\mathbb{I}_{\{a<x<b\}}$, $x\in\mathbb{R}$.
	Given the current position $x_n$ and the parameters $(V_n, v_n)$ from a uniform distribution on $(0,1)\times\{-1,+1\}$, the $\lambda(\cdot)$ function becomes $\lambda(t)=\left(U'(x_n+v_nt)\cdot v_n\right)_+$ $=\begin{cases} 0, & x_n+v_nt\in (a,b)\\ \infty, & \textnormal{otherwise} \end{cases}.$ In this case,
	$\tau$ depends on only $x_n, v_n$ as
	\begin{equation*}
		\tau=\begin{cases} b-x_n, & \textnormal{ for }\; v_n=1 \\ x_n-a, & \textnormal{ for }\; v_n=-1 \end{cases}.
	\end{equation*}
Given $x_n$, the update point $x_{n+1}$ becomes $\frac{b+x_n}{2}$ with probability 0.5 (when $v_n=1$) and $\frac{a+x_n}{2}$ with probability 0.5 (when $v_n=-1$). To see the $\textnormal{Unif}(a,b)$ distribution is invariant for this chain, denote the kernel of the proposed Markov chain as
\begin{equation*}
	K(x_n,x_{n+1})=\frac{1}{2}\delta_{\frac{x_n+b}{2}}(x_{n+1})+\frac{1}{2}\delta_{\frac{x_n+a}{2}}(x_{n+1}),
\end{equation*}
where $\delta_x(\cdot)$ is the Dirac mass located at $x$. Since 
\begin{equation*}
\begin{aligned}
	&\int_{-\infty}^{\infty}\pi(x_n)K(x_n,x_{n+1})dx_n=\frac{1}{2(b-a)}\int_a^b\delta_{\frac{x_n+b}{2}}(x_{n+1})dx_n+\frac{1}{2(b-a)}\int_a^b\delta_{\frac{x_n+a}{2}}(x_{n+1})dx_n \\
	&=\frac{1}{2(b-a)}\int_{\frac{a+b}{2}}^b\delta_z(x_{n+1})2dz+\frac{1}{2(b-a)}\int_a^{\frac{a+b}{2}}\delta_z(x_{n+1})2dz=\frac{1}{b-a}\mathbb{I}_{\{a<y<b\}}=\pi(x_{n+1}),
\end{aligned}	
\end{equation*}
we see $\pi$ is invariant (stationary) for the proposed kernel $K(x_n,x_{n+1})$.
\end{example}


\begin{example}[Beta distribution] \label{example:beta}
We consider an example of $Beta(\alpha,\beta)$ distribution with density $\pi(x)=\frac{x^{\alpha-1}(1-x)^{\beta-1}}{B(\alpha,\beta)}$, $x\in(0,1)$. The negative logarithm of this density becomes $U(x)=-(\alpha-1)\log x-(\beta-1)\log(1-x)+\log B(\alpha,\beta)$. We consider here the case $\alpha>1$ and $\beta>1$, implying $U$ is convex. The other cases are presented in the appendix in detail. 
Given $V_n, v_n$ and $x_n$, we compute $\tau$, differentiating between positive and negative velocity $v_n$. 
\begin{itemize}[leftmargin=*]
	\item Case $v_n=1$. Denote $t^*=\textnormal{arg}\:\underset{t\geq 0}{\min}\: U(x_n+t)=\left(\frac{1}{1+\frac{\beta-1}{\alpha-1}}-x_n\right)_+$. Then $U'(x_n+t)<0$ for $t<t^*$ and $U'(x_n+t)>0$ for $t>t^*$. $\tau$ solves
	\begin{equation*}
		-\log V_n=\int_0^{\tau} U'(x_n+t)_+dt=\int_{t^*}^\tau U'(x_n+t)dt=U(x_n+\tau)-U(x_n+t^*),
	\end{equation*}
	$\tau\in[t^*,1-x_n]$, which we solve for numerically though a line search.
	\item Case $v_n=-1$. Denote $t^*=\textnormal{arg}\:\underset{t\geq 0}{\min}\: U(x_n-t)=\left(x_n-\frac{1}{1+\frac{\beta-1}{\alpha-1}}\right)_+$. Then $U'(x_n-t)>0$ for $t<t^*$ and $U'(x_n-t)<0$ for $t>t^*$. $\tau$ solves
	\begin{equation*}
	\begin{aligned}
	 	-\log V_n &=\int_0^{\tau}(-U'(x_n-t))_+dt=\int_{t^*}^{\tau}-U'(x_n-t)dt \\
	 	&=\int_{-t^*}^{-\tau}U'(x_n+s)ds= U(x_n-\tau)-U(x_n-t^*),
	 \end{aligned}
	\end{equation*}
	$\tau\in[t^*, x_n]$. Using the proposed sampler, we simulate several Beta distributions with varying $\alpha$ and $\beta$ parameters in Figure \ref{fig:beta:histograms}.
\end{itemize}
\end{example}

\begin{figure}[h]
\centering
\includegraphics[width=0.8\textwidth]{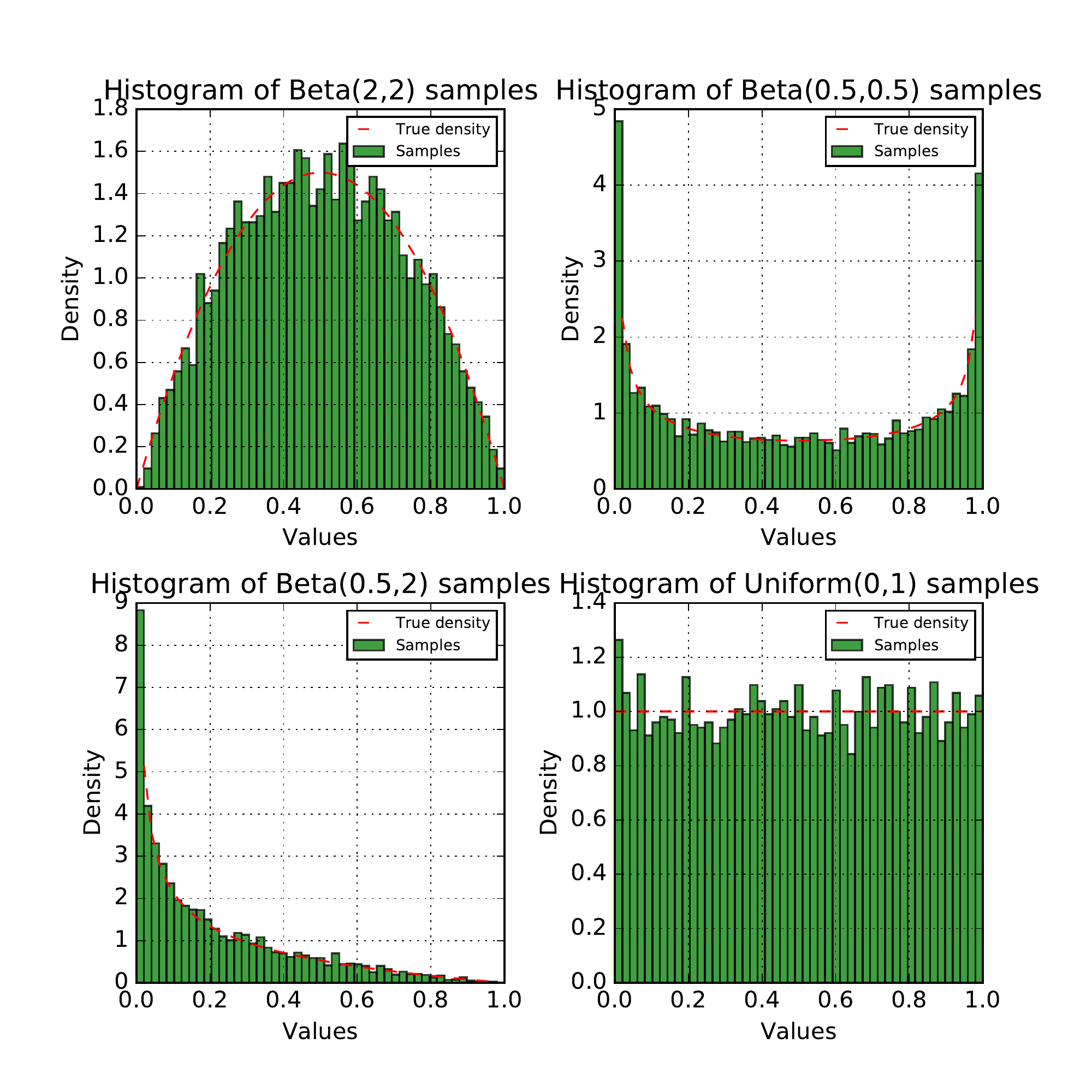}
\caption{Simulating Beta distributions.}
\label{fig:beta:histograms}
\end{figure}

Treating $v_n$ as fixed and $\tau$ as a function of a random variable $V_n\sim \textnormal{Unif}(0,1)$, $\tau$ becomes the first arrival time of an inhomogeneous Poisson process. $\Pi$ with intensity function $\lambda(t)$. The survival function of $\tau$ is then
\begin{equation*}
	\mathbb{P}\{\tau>t\}=\mathbb{P}\left\{\Pi\cap[0,t]=\emptyset\right\}=e^{-\int_0^t\lambda(s)ds}.
\end{equation*}
Our sampler requires having a sample from the distribution of $\tau$. In general it is not possible to obtain an analytic expression for $\tau$. \cite{bouchard2017bouncy, bierkens2016zig} describe the numerical ways to get a sample from the distribution of $\tau$ by e.g.~using adaptive thinning methods for sampling the first arrival time. We assume for now that we can find $\tau$ that solves \eqref{eq:tau:equation}.

\begin{remark1}\label{remark:discrete:vs:continuous:time}
After running the proposed discrete-time MCMC, the collected samples $x_n$, $n\geq 1$, are marginally from $\pi$. The previous works construct continuous-time Markov chains, providing a realization of $x(t)$ over the time interval $(0,T)$, where $T$ is the total time the Markov chain moves. They estimate the integral of interest $\int_{\mathbb{R}^d}\phi(x)\pi(dx)$ using $\frac{1}{T}\int_0^T\phi(x(t))dt$, which might not be tractable. In that case, to estimate the integral they evaluate $x(t)$ at equally-spaced and discrete time points $i\cdot \Delta t$, $i=0,\ldots,\lfloor T/\Delta t \rfloor=I$, obtaining an estimator $\frac{1}{I}\sum_{i=0}^I\phi(x(i\cdot \Delta t))$. This requires specifying a parameter $\Delta t$.
\end{remark1}

\section{The invariance of the target density} \label{sec:invariance}

We prove the target density is invariant with respect to the kernel corresponding to the proposed Markov chain. It suffices to prove the invariance in the univariate case, which is the setting of this section. Then the theory of unifying hit-and-run Markov chain method implies that the proposed MCMC is invariant for the target distributions in $d$-dimensions as well \citep{andersen2007hit}. 

Given a current position $x_n\in\mathbb{R}$, the chain moves to the right ($v_n=1$) or left ($v_n=-1$) with probability $\frac{1}{2}$.
\begin{itemize}[leftmargin=*]
	\item If the chain goes right, the next position is $x_n+\frac{\tau}{2}=x_{n+1}>x_n$, where $\tau$ solves $-\log V_n = \int_0^\tau(U'(x_n+t))_+dt$.
	\item If the chain goes left, the next position is $x_n-\frac{\tau}{2}=x_{n+1}<x_n$, where $\tau$ solves $-\log V_n=\int_0^{\tau}(-U'(x_n-t))_+dt$.
\end{itemize}
The kernel for the proposed Markov chain becomes
\begin{equation} \label{eq:kernel:general:U}
\begin{aligned}
	K(x_n,x_{n+1}) = &(U'(2x_{n+1}-x_n))_+\cdot e^{-\int_0^{2(x_{n+1}-x_n)}(U'(x_n+t))_+dt}\cdot \mathbb{I}_{\{x_{n+1}>x_n\}}\\
	+&(-U'(2x_{n+1}-x_n))_+\cdot e^{-\int_0^{2(x_n-x_{n+1})}(-U'(x_n-t))_+dt}\cdot \mathbb{I}_{\{x_{n+1}<x_n\}}.
\end{aligned}
\end{equation}

\begin{theorem}  \label{thm:stationarity}
The density proportional to $e^{-U(x)}$ is stationary for the kernel in \eqref{eq:kernel:general:U}  .
\end{theorem}

\begin{proof}

Without the loss of generality we can assume $U(0)=0$. Since
\begin{equation*}
	\begin{aligned}
		\int_{-\infty}^{\infty}K(x_n,x_{n+1})\pi(x_n)dx_n
		&=\int_{-\infty}^{x_{n+1}}(U'(2x_{n+1}-x_n))_+\cdot e^{-\int_0^{2(x_{n+1}-x_n)}(U'(x_n+t))_+dt}\pi(x_n)dx_n\\
		&+\int_{x_{n+1}}^{\infty}(-U'(2x_{n+1}-x_n))_+\cdot e^{-\int_0^{2(x_n-x_{n+1})}(-U'(x_n-t))_+dt}\pi (x_n)dx_n,
	\end{aligned}
\end{equation*}
	differentiating the above with respect to $x_{n+1}$ we get
\begin{equation*}
\begin{aligned}
	&\pi(x_{n+1})(U'(x_{n+1}))_+-\pi(x_{n+1})(-U'(x_{n+1}))_+\\
	&\qquad+\int_{-\infty}^{x_{n+1}}\frac{d}{dx_{n+1}}\left((U'(2x_{n+1}-x_n))_+\cdot e^{-\int_0^{2(x_{n+1}-x_n)}(U'(x_n+t))_+dt}\right)\pi(x_n)dx_n \\
	&\qquad+\int_{x_{n+1}}^{\infty}\frac{d}{dx_{n+1}}\left((-U'(2x_{n+1}-x_n))_+\cdot e^{-\int_0^{2(x_n-x_{n+1})}(-U'(x_n-t))_+dt}\right)\pi(x_n)dx_n.
\end{aligned}
\end{equation*}
To prove the invariance of density $\pi(x)\:\propto\:e^{-U(x)}$ with respect to the given kernel, we need to prove
\begin{equation*}
	\pi(x_{n+1})=\int_{-\infty}^{\infty}K(x_n, x_{n+1})\pi(x_n)dx_n.	
\end{equation*}
Since $(U'(x))_+-(-U'(x))_+=U'(x)$ and $\pi'(x)=-U'(x)\pi(x)$, it suffices to show
\begin{equation} \label{eq:main:theorem:suff:condition}
\begin{aligned}
	&-2U'(x_{n+1})e^{-U(x_{n+1})} \\
	&=\int_{-\infty}^{x_{n+1}}\frac{d}{dx_{n+1}}\left((U'(2x_{n+1}-x_n))_+\cdot e^{-\int_0^{2(x_{n+1}-x_n)}(U'(x_n+t))_+dt}\right)e^{-U(x_n)}dx_n \\
	&\qquad +\int_{x_{n+1}}^{\infty}\frac{d}{dx_{n+1}}\left((-U'(2x_{n+1}-x_n))_+\cdot e^{-\int_0^{2(x_n-x_{n+1})}(-U'(x_n-t))_+dt}\right)e^{-U(x_n)}dx_n.
\end{aligned}
\end{equation}
Denoting $(U'(x))_+=u_1(x)$ and $-(-U'(x))_+=u_2(x)$, the RHS of \eqref{eq:main:theorem:suff:condition} equals for all $x_{n+1}$ to
\begin{equation*}
\begin{aligned}
	&\int_{-\infty}^{x_{n+1}}\frac{d}{dx_{n+1}}\left(u_1(2x_{n+1}-x_n)\cdot e^{-\int_0^{2(x_{n+1}-x_n)}u_1(x_n+t)dt}\right) e^{-U(x_n)}dx_n\\
	&\qquad +\int_{x_{n+1}}^{\infty}\frac{d}{dx_{n+1}}\left(-u_2(2x_{n+1}-x_n)\cdot e^{-\int_0^{2(x_n-x_{n+1})}(-u_2(x_n-t))dt}\right)e^{-U(x_n)}dx_n \\
	&=\int_{-\infty}^{x_{n+1}}\frac{d}{dx_{n+1}}\left(u_1(2x_{n+1}-x_n)\cdot e^{-\int_0^{2(x_{n+1}-x_n)}u_1(x_n+t)dt}\right)e^{-U(x_n)}dx_n \\
	&\qquad +\int_{-\infty}^{x_{n+1}}\frac{d}{dx_{n+1}}\left(-u_2(x_n)\cdot e^{-\int_0^{2(x_{n+1}-x_n)}(-u_2(2x_{n+1}-x_n-t))dt}\right)e^{-U(2x_{n+1}-x_n)}dx_n,
\end{aligned}
\end{equation*}
where in the second equality we did the change of variables $x_n\rightarrow 2x_{n+1}-x_n$ for the second integral. 
Since $U(x_n)=\int_0^{x_n}u_1(t)dt+\int_0^{x_n}u_2(t)dt$, the expression above becomes equal to
\begin{equation*}
\begin{aligned}
    &2\int_{-\infty}^{x_{n+1}}(u_1'(2x_{n+1}-x_n)-u_1^2(2x_{n+1}-x_n))e^{-\int_0^{2x_{n+1}-x_n}u_1(t)dt-\int_0^{x_n} u_2(t)dt}dx_n \\
	&\qquad +2\int_{-\infty}^{x_{n+1}}(-u_2'(x_n)+u_2^2(x_n))e^{-\int_0^{x_n}u_2(t)dt -\int_0^{2x_{n+1}-x_n}u_1(t)dt}dx_n \\
	&=2\int_{-\infty}^{x_{n+1}}(u_1'(2x_{n+1}-x_n)-u_2'(x_n))e^{-\int_0^{2x_{n+1}-x_n}u_1(t)dt-\int_0^{x_n} u_2(t)dt}dx_n \\
	&\qquad +2\int_{-\infty}^{x_{n+1}}(-u_1^2(2x_{n+1}-x_n)+u_2^2(x_n))e^{-\int_0^{2x_{n+1}-x_n}u_1(t)dt -\int_0^{x_n} u_2(t)dt}dx_n. 
\end{aligned}	
\end{equation*}
Denoting $\kappa(x_n)=-u_1(2x_{n+1}-x_n)-u_2(x_n)$ and $g(x_n)=e^{-\int_0^{2x_{n+1}-x_n}u_1(t)dt -\int_0^{x_n} u_2(t)dt}$, the above becomes
\begin{equation*}
\begin{aligned}
	&2\int_{-\infty}^{x_{n+1}}\kappa'(x_n)g(x_n)dx_n+2\int_{-\infty}^{x_{n+1}}\kappa(x_n)g'(x_n)dx_n
	=2\int_{-\infty}^{x_{n+1}}(\kappa(x_n)g(x_n))'dx_n \\
	&=2\kappa(x_{n+1})g(x_{n+1})=-2(u_1(x_{n+1})+u_2(x_{n+1}))e^{-\int_0^{x_{n+1}}u_1(t)dt-\int_0^{x_{n+1}}u_2(t)dt}\\
	&=-2U'(x_{n+1})e^{-U(x_{n+1})},
\end{aligned}
\end{equation*}
finishing the proof.
\end{proof}


\begin{remark1}
	In order to apply hit-and-run framework and get the invariance of our target distribution under the proposed chain, we do not necessarily need to draw $v_n$ from a uniform distribution on $\mathbb{S}^{d-1}$ at each update step (or, for all $n$). It suffices to keep $v_n$ along the same line for a fixed number of steps and only multiply it with $+1$ or $-1$ with probabilities 0.5 each. As long as the number of steps we keep the velocity in the same line is not dependent on the chain positions so far, the target distribution stays invariant.
\end{remark1}


\section{Easier sampling via decomposition} \label{sec:sampling:via:decomposition}

This section presents a slightly different, but similar in nature, sampler in the univariate case ($d=1$) which is easier to implement when a particular representation of the log-density is available.
Assume that $U$ can be decomposed as the sum $U=U_1+U_2$ with $U_1:\mathbb{R}\rightarrow\mathbb{R}$ an increasing function and $U_2:\mathbb{R}\rightarrow\mathbb{R}$ a decreasing function. Every bounded variation function on a real line can be decomposed in such a way. However, we do not always have the explicit forms for $U_1$ and $U_2$. If $U_1$ and $U_2$ can be computed efficiently, we introduce a new sampler that avoids sampling of the first arrival times of the Poisson process.

In terms of the iterated random functions framework, the family $\{f_{\theta}\mid\theta\in \Theta\}$ of transition functions corresponding to this Markov chain is different compared to the sampler in Section \ref{sec:sampler}. Even though the form \eqref{eq:irf} stays the same, $\tau$ as a function of $v_n, V_n$ and $x_n$ changes. The probability measure $\mu$ on $\Theta=(0,1)\times \{-1,1\}$ stays the same.
To describe the amount of move, we differentiate between two cases depending if $v_n$ is positive or negative. For the new Markov chain $\tau$ becomes
\begin{equation} \label{eq:tau:decomposition}
	\tau = \begin{cases}
 	U_1^{-1}(U_1(x_n)-\log V_n)-x_n & \textnormal{ for }\; v_n=1 \\
 	x_n-U_2^{-1}(U_2(x_n)-\log V_n) & \textnormal{ for }\; v_n=-1
 \end{cases}	
\end{equation}

Combining the two cases above, the proposed Markov chain has the following kernel
\begin{equation} \label{eq:kernel:via:decomposition}
\begin{aligned}
	K(x_n,x_{n+1})&=\mathbb{I}_{\{x_{n+1}>x_n\}}\left(U_1'(2x_{n+1}-x_n)e^{-U_1(2x_{n+1}-x_n)+U_1(x_n)}\right) \\
	&+\mathbb{I}_{\{x_{n+1}<x_n\}}\left(-U_2'(2x_{n+1}-x_n)e^{-U_2(2x_{n+1}-x_n)+U_2(x_n)}\right).
\end{aligned}
\end{equation}

\begin{theorem} \label{thm:stationarity:decomposition}
	The univariate target density $\pi(x)\:\propto\: e^{-U(x)}=e^{-U_1(x)-U_2(x)}$, $x\in\mathbb{R}$, where $U_1$ is an increasing function and $U_2$ is a decreasing function, is stationary for the kernel proposed in \eqref{eq:kernel:via:decomposition}.
\end{theorem}
The proof of this theorem is similar to the proof of Theorem \ref{thm:stationarity} and is given in Appendix \ref{app:proofs}.

We illustrate the computational simplification of the proposed Markov chain that uses the decomposition in two examples. 

\begin{example}[Mixture of two Gaussians]
We sample from a mixture of two univariate Gaussians using the Markov chain proposed in this section. The target probability density has the weights $w_1>0$ on $\mathcal{N}(\mu_1,\sigma^2)$ and $w_2=1-w_1$ on $\mathcal{N}(\mu_2,\sigma^2)$, $\mu_2>\mu_1$, equalling
\begin{equation*}
\begin{aligned}
	\pi(x) & \:\propto\: w_1\cdot e^{-\frac{1}{2\sigma^2}(x-\mu_1)^2}+w_2\cdot e^{-\frac{1}{2\sigma^2}(x-\mu_2)^2}\\
	&=\exp\left(-\frac{1}{2\sigma^2}(x-\mu_1)^2+\ln\left(w_1+w_2\cdot e^{\frac{1}{2\sigma^2}\left(2x(\mu_2-\mu_1)+\mu_1^2-\mu_2^2\right)}\right)\right).
\end{aligned}
\end{equation*}
Thus, $\pi(x)=e^{-U_1(x)-U_2(x)}$ for an increasing function $U_1(x)=\frac{1}{2\sigma^2}(x-\mu_1)^2\mathbb{I}_{\{x>\mu_1\}}$ and a decreasing function $U_2(x)=\frac{1}{2\sigma^2}(x-\mu_1)^2\mathbb{I}_{\{x<\mu_1\}}-\ln\left(w_1+w_2\cdot e^{\frac{1}{2\sigma^2}\left(2x(\mu_2-\mu_1)+\mu_1^2-\mu_2^2\right)}\right)$. Figure \ref{fig:gaussian:mixture} illustrates the validity of this approach.
\end{example}

\begin{example}[Beta distribution via decomposition] \label{example:beta:decomposed} 
	We describe the new Markov chain in the case of $Beta(\alpha,\beta)$ distribution when both $\alpha$ and $\beta$ are greater than 1. Other cases are considered in Appendix \ref{app:examples}.
	The negative log-density of $Beta(\alpha,\beta)$ decomposes as
	$U(x)=U_1(x)+U_2(x)$, $x\in(0,1)$, for an increasing function $U_1(x)=-(\beta-1)\log(1-x)$ with $U_1^{-1}(y)=1-e^{-\frac{y}{\beta-1}}$ and a decreasing function $U_2(x)=-(\alpha-1)\log x$ with $U_2^{-1}(y)=e^{-\frac{y}{\alpha-1}}$. Given the random parameters $(V_n, v_n)\in (0,1)\times\{-1,1\}$, $\tau$ becomes
	\begin{equation*}
		\tau = \begin{cases} U_1^{-1}(U_1(x_n)-\log V_n)-x_n\in (0,1-x_n) & \textnormal{ for } v_n=1 \\ x_n-U_2^{-1}(U_2(x_n)-\log V_n)\in (0,x_n) & \textnormal{ for } v_n=-1
		\end{cases}.
	\end{equation*}
	The update point becomes $x_n+\frac{1}{2}\tau$ for $v_n=1$ and  $x_n-\frac{1}{2}\tau$ for $v_n=-1$.
\end{example}

\begin{figure}[h]
\centering
\includegraphics[width=0.6\textwidth]{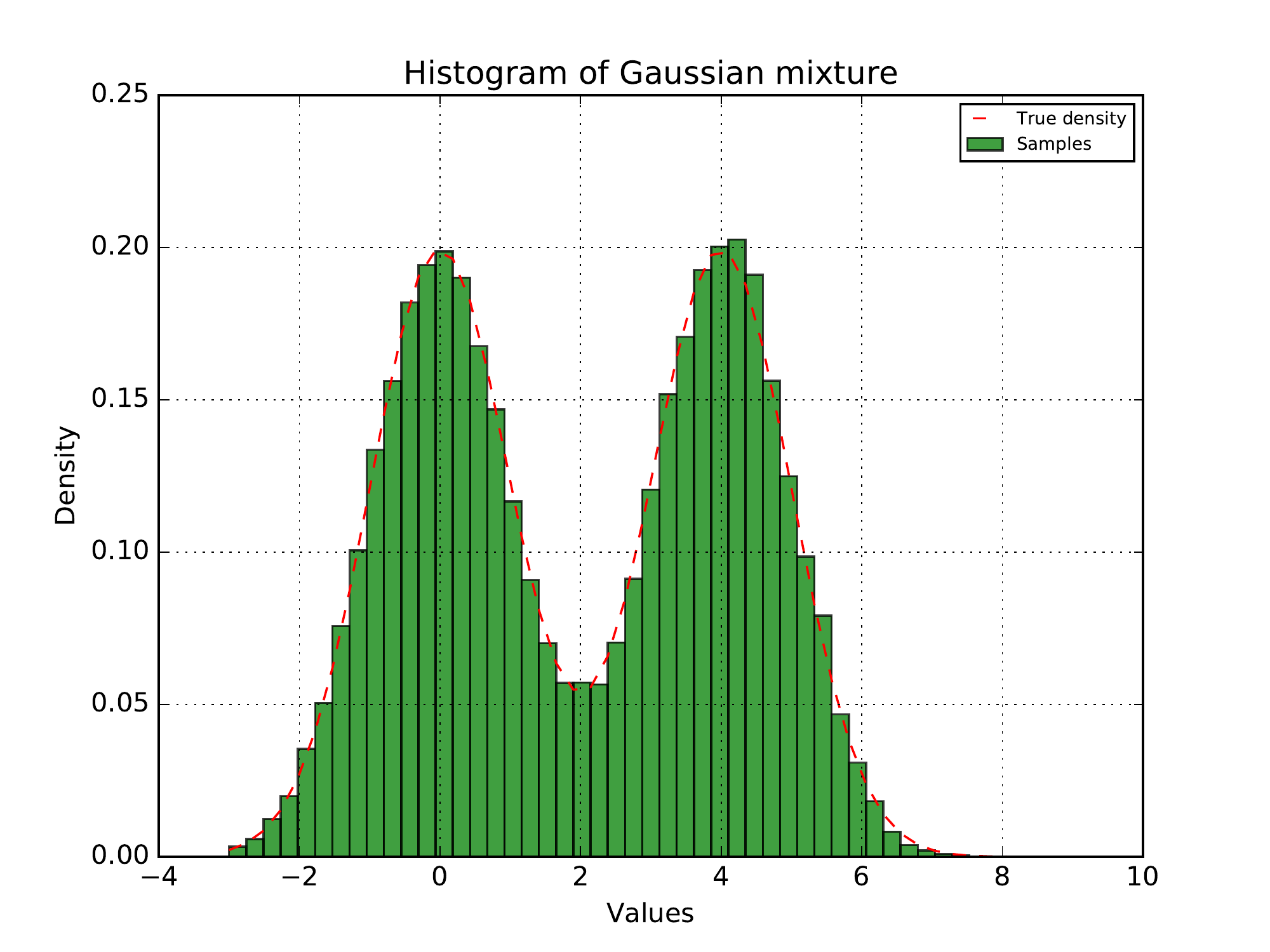}
\caption{Simulating mixture of two Gaussians $\frac{1}{2}\mathcal{N}(0,1)+\frac{1}{2}\mathcal{N}(4,1)$.}
\label{fig:gaussian:mixture}
\end{figure}

\begin{remark1}
	Having $v_n$ fixed at 1 or -1 and taking into account randomness in $V_n\sim\textnormal{Unif}(0,1)$, $\tau$ defined in \eqref{eq:tau:decomposition} represents the first arrival time of a corresponding Poisson process. We differentiate between two Poisson processes depending whether $v_n$ is positive or negative. 
	\begin{itemize}
		\item For $v_n=1$, $\tau$ is the first arrival time of the PP with intensity function $\lambda(t)=U_1'(x_n+t)$
		\item For $v_n=-1$, $\tau$ is the first arrival time of the PP with intensity function $\lambda(t)=U_2'(x_n-t)$.
	\end{itemize}  
\end{remark1}

\section{Selective sampler after LASSO} \label{sec:application}

\subsection{Background}

Practitioners often perform some algorithms on their data (selection) and choose their models and hypothesis to test (inference) upon seeing the outcomes of those algorithms. This problem, colloquially called ``data snooping'' or ``cherry picking,'' does not fit into the framework of classical statistics which assumes the models and hypothesis are fixed in advanced, before seeing the data. Ignoring the selection leads of inflated $p$-values and confidence intervals with poor coverage. Unless properly adjusted, the classical $p$-values and confidence intervals are no longer valid.
Selective inference solves the problem of providing valid $p$-values and confidence intervals after performing a model selection algorithm on the data. The field started with the works of \cite{lee2013exact, lee_screening, sequential_post_selection}, where the authors constructed a pivotal quantity that is $\textnormal{Unif}(0,1)$ distributed after selection. Using the constructed test statistic, they provide valid $p$-values and confidence intervals for the variables selected using LASSO, marginal screening and sequential procedures, respectively.

Since their approach can suffer from low statistical power, the randomized selective inference, starting with \cite{tian2015selective}, solves a randomized model selection problem to select variables of interest. Then it provides inference for only the selected coefficients by calibrating the corresponding test statistics using the conditional distribution of the data, where conditioning is on the selected model. This randomized approach will be our focus here. For further background on selective inference see the review papers \cite{taylor2015statistical, bi2017inferactive}.

Before describing the use of our sampler here, we illustrate the randomized selective inference approach through a randomized LASSO example.
The data consists of the design matrix $X\in\mathbb{R}^{n\times p}$ with rows $x_i\in\mathbb{R}^p$, $i=1,\ldots,n$, and the response vector $y=(y_1,\ldots, y_n)\in\mathbb{R}^n$. We assume $(x_i,y_i)$, $i=1,\ldots, n$, are i.i.d.~from a distribution $\mathbb{F}$.
Given data $(X,y)$ and a sample $\omega\sim\mathbb{G}$ with density $g$ from a pre-specified distribution $\mathbb{G}$, the randomized LASSO solves the following convex objective
\begin{equation} \label{eq:randomized:lasso}
	\underset{\beta\in\mathbb{R}^p}{\textnormal{minimize}}\:\frac{1}{2}\|y-X\beta\|_2^2+\lambda\|\beta\|_1+\frac{\epsilon}{2}\|\beta\|_2^2-\omega^T\beta, \;\; ((X,y),\omega)\sim\mathbb{F}^n\times\mathbb{G},
\end{equation}
where $\lambda$ is the $\ell_1$-penalty level and $\epsilon>0$ is a ridge term, usually taken to be a small constant, ensuring the solution of the above objective exists \citep{selective_sampler}. Denote the solution of  \eqref{eq:randomized:lasso} as $\hat{\beta}=\hat{\beta}(X,y,\omega)$. Since the randomized LASSO objective produces a sparse solution we denote the non-zero coordinates of $\hat{\beta}$ as $E$, writing $\hat{\beta}=\begin{pmatrix} \hat{\beta}_E \\ 0 \end{pmatrix}$. The set $E$ determines the set of selected variables which are of interest as related to our response. Given the selected set $E$, we decide to report the $p$-values and confidence intervals for the populations coefficients $\beta_{E,j}*$, $j\in E$, corresponding to the selected coefficients, where 
\begin{equation*}
	\beta_E^*
	=\left(\mathbb{E}_{\mathbb{F}^n}[X_E^TX_E]\right)^{-1}\mathbb{E}_{\mathbb{F}^n}[X_E^Ty]
	=\left(\mathbb{E}_{\mathbb{F}}[x_{i,E}x_{i,E}^T]\right)^{-1}\mathbb{E}_{\mathbb{F}}[x_{i,E}y_i] 
\end{equation*}
with $X_E$ denoting the sub-matrix of $X$ consisting of columns in $E$.
Assuming $E$ is fixed and not chosen based on data, we can construct the $p$-values and confidence intervals for each of the $\beta_{E,j}^*$, $j\in E$, coordinates using the asymptotic normality of the least-squares estimator $\bar{\beta}_E=(X_E^TX_E)^{-1}X_E^Ty$. 
Since $E$ is random, i.e.~chosen based on data, the classical (which we also call ``naive'') $p$-values and confidence intervals based on normal quantiles will no longer be valid. Using selective inference approach, we also base inference using $\bar{\beta}_E$ as a test statistic; however, its distribution under the null is no longer normal but normal distribution conditional on the event that the randomized LASSO chose set $E$. Precisely, the conditioning is on set that $(X,y,\omega)$ belongs to the set
\begin{equation} \label{eq:conditioning:event}
	\left\{(X',y',\omega'):\hat{\beta}(X',y',\omega')_{-E}\neq 0, \hat{\beta}(X',y',\omega')_{-E}=0\right\}.
\end{equation}

To get a handle on the corresponding distribution of $\bar{\beta}_E$ conditional on the event \eqref{eq:conditioning:event}, we rely on the change of measure approach of \cite{selective_sampler}. Note that the conditioning event is written in terms of the whole data $(X,y)$ and randomization vector $\omega$. However, the LASSO solution depends on $\omega$ and a data vector $D=\begin{pmatrix} \bar{\beta}_E \\ X_{-E}^T(y-X_E\bar{\beta}_E)\end{pmatrix}$ as follows. Denoting with $u_{-E}$ is the inactive subgradient of $\ell_1$-penalty in \eqref{eq:randomized:lasso}, we write the KKT conditions of the randomized LASSO objective as 
\begin{equation} \label{eq:KKT}
	\omega = \omega(D,\hat{\beta}_E,u_{-E})
	=-\begin{pmatrix} X_E^TX_E & 0 \\ X_E^TX_{-E} & I_{p-|E|}\end{pmatrix}D +\begin{pmatrix}
		X_E^TX_E+\epsilon I_{|E|} \\ X_{-E}^TX_E \end{pmatrix}\hat{\beta}_E+\begin{pmatrix}\lambda s_E \\ u_{-E}\end{pmatrix}, 
\end{equation}
with the constraints $\|u_{-E}\|_{\infty}<\lambda$ and $\textnormal{sign}(\hat{\beta}_E)=s_E$, where $I_k$ denotes the identity matrix of the dimension $k$. We treat vector $s_E$ as fixed at its observed value, i.e.~we condition on the sign vector to simplify the selection region and get the log-concave selective density.
Under mild conditions on $\mathbb{F}$, vector $D$ is asymptotically normal $\mathcal{N}(\mu_D,\Sigma_D)$ treating $E$ as fixed and not chosen based on data. We treat the matrices in the randomization reconstruction in \eqref{eq:KKT} as fixed which is justified by the Strong law of large numbers. Then the solution of LASSO depends only on $D$ and $\omega$. Instead of writing the selective density on $D$ and $\omega$ with a complicated set of affine constraints describing the set that solving LASSO on $(D,\omega)$ gives the set $E$, \cite{selective_sampler} write the selective density in terms of $(D,\hat{\beta}_E, u_{-E})$ with simple constraints on so called optimization variables $\hat{\beta}_E$ and $u_{-E}$. Thus the selective density on $(D,\hat{\beta}_E,u_{-E})$ becomes
\begin{equation} \label{eq:selective:density}
	\phi_{(\mu_D,\Sigma_D)}(D)\cdot g(\omega(D,\hat{\beta},u_{-E}))
\end{equation}
with constraints $\textnormal{sign}(\hat{\beta}_E)=s_E$ and $\|u_{-E}\|_{\infty}<\lambda$, where $\phi_{(\mu,\Sigma)}(\cdot)$ denotes the density of normal $\mathcal{N}(\mu,\Sigma)$ distribution. Note that the constraints on the optimization variables are restricting $\hat{\beta}_E$ to a particular orthant and $u_{-E}$ to a cube. The density in \eqref{eq:selective:density} is much simpler to sample from than the density on $(D,\omega)$ with the constraints on both $D$ and $\omega$.

To construct $p$-values and confidence intervals for all coordinates $\beta_E^*$, $j\in E$, we first sample using our proposed Markov chain only the optimization variables from the density in \eqref{eq:selective:density} while keeping the data vector $D$ fixed at its observed value. Then we do importance sampling as described in \cite{bootstrap_mv}. This allows us to run the MCMC sampler only once in order to construct all $|E|$ $p$-values and confidence intervals.
To prove validity of this approach, the distribution $\mathbb{G}$ is usually taken to be Gaussian or heavy-tailed such as Laplace or logistic, thus log-concave \citep{tian2015selective, bootstrap_mv}. We sample the optimization variables from the constrained density in \eqref{eq:selective:density} with data vector $D$ fixed at its observed value using our proposed sampler.

\subsection{Sampling from a truncated log-concave density} \label{sec:truncated:log:concave}

We describe sampling from a truncated log-concave density using our proposed sampler in Section \ref{sec:sampler}. Let the density be proportional to $e^{-U(x)}$, $x\in D\subset\mathbb{R}^d$, with $U$ a convex function. Given the current Markov chain position $x_n$ and the random parameters $(V_n,v_n)$, computing $\tau$ consists of the following.
\begin{enumerate}[leftmargin=*]
	\item Compute $t_{\min}=\min\{t\geq 0:x_n+v_nt\in D\}$ and $t_{\max}=\max\{t\geq 0:x_n+v_nt\in D\}$. 
	\item Solve $t^*=\textnormal{arg}\:\underset{t\in[t_{\min}, t_{\max}]}{\min}\: U(x_n+v_nt)$. 
	\item If $U(x_n+v_nt_{\max})-U(x_n+v_nt^*)+\log V_n>0$, then $\tau\in(t^*, t_{\max})$ solves
	\begin{equation*}
		-\log V_n=U(x_n+v_n\tau)-U(x_n+v_nt^*).	
	\end{equation*}
	Otherwise, set $\tau = t_{\max}$.
\end{enumerate}

We apply the algorithm above to sample from an univariate truncated Gaussian distribution (Figure \ref{fig:truncated:gaussian}).

\begin{figure}[h]
\centering
\includegraphics[width=0.7\textwidth]{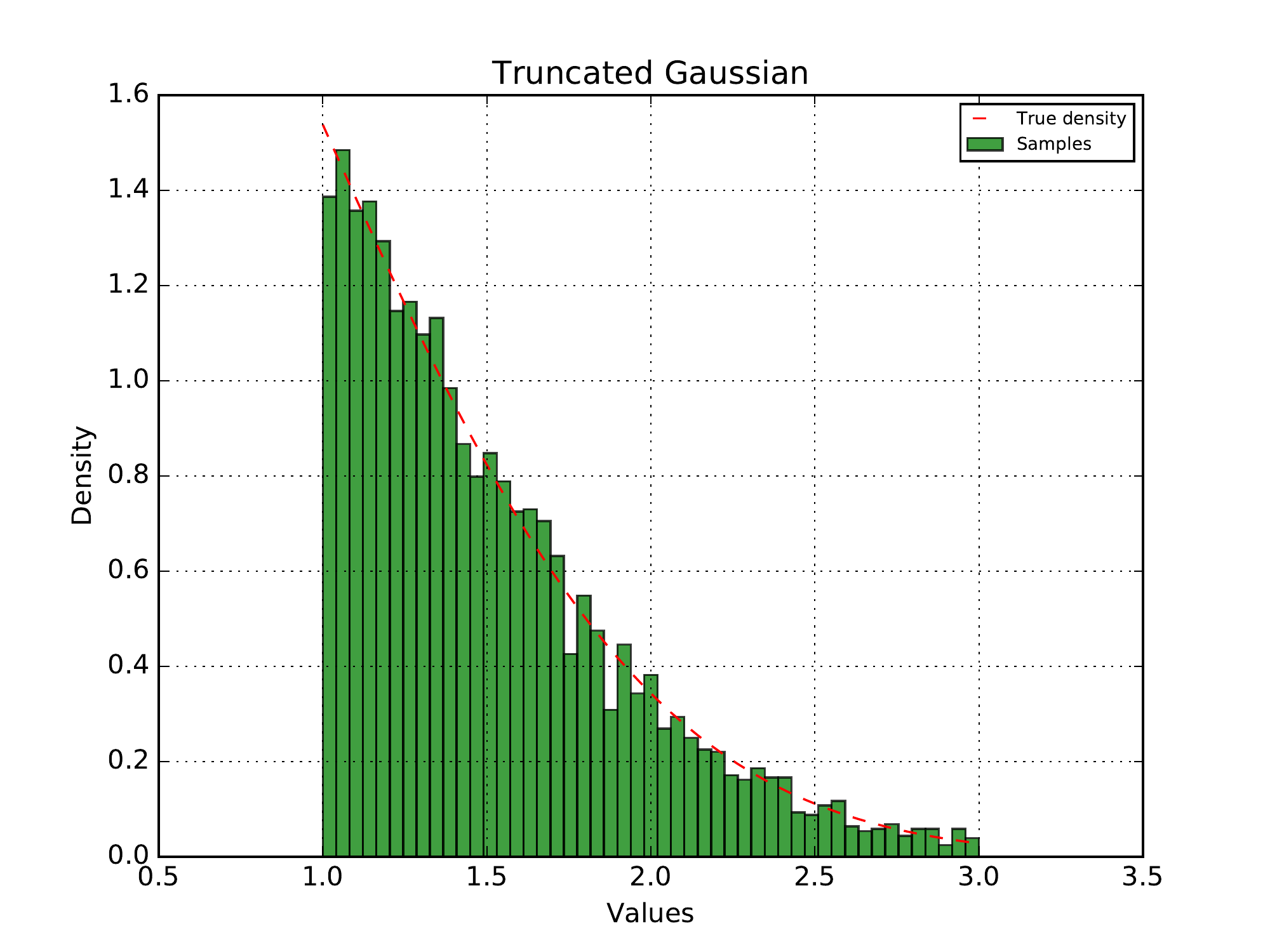}
\caption{Simulating Gaussian $\mathcal{N}(0,1)$ distribution truncated to $[1,3]$.}
\label{fig:truncated:gaussian}
\end{figure}

\subsection{Selective inference results }

We apply the sampler outlined above for sampling the optimization variables from the selective density in \eqref{eq:selective:density}. In this example the constraint set $D$ is a product of orthants and cubes, hence we have explicit solutions for $t_{\min}$ and $t_{\max}$.
Before presenting the results, let us describe the simulation details. The design matrix $X$ is generated from equi-correlated model, i.e.~the rows $x_i$, $i=1,\ldots, n$, are generated i.i.d.~from $\mathcal{N}(0,\Sigma)$, $\Sigma_{i,j}=\begin{cases} 1 & i=j \\ \rho & i\neq j\end{cases}$. We further normalize the columns of $X$ to have norm 1. The response $y$ is generated as $\mathcal{N}(0,I_n)$ and is independent of $X$. Note that we are in the global null setting, i.e.~$\beta_E^*$ is a zero vector for any selected set $E$. Denoting the empirical standard deviation of $y$ as $\hat{\sigma}$, we take $\epsilon=\frac{\hat{\sigma}^2}{\sqrt{n}}$ and the randomization scale $s=\frac{1}{2}\hat{\sigma}$. In the case of Gaussian randomization, $\omega\sim\mathcal{N}(0,s^2)$ and in the case of Laplace randomization $\omega\sim \textnormal{Laplace}(\textnormal{loc}=0,\textnormal{scale}=s)$. 
We set the dimensions to $n=100, p=40$, the correlation among the predictors $\rho=0.3$ and the penalty level $\lambda=1.4$ in simulations.

After generating the data as above, randomized LASSO objective in \eqref{eq:randomized:lasso} gives an active set $E$. Then the report consists only the $p$-values and confidence intervals corresponding to this selected set. We construct the $p$-values corresponding to this set using both the proposed sampler (MC via IRF) and the projected Langevin sampler. The latter has been used for sampling from truncated log-concave densities and its theoretical guarantees are presented in \cite{bubeck2015sampling}. We include the naive $p$-values and confidence intervals as well, further emphasizing the point that they do not preserve the targeted type I-error rate and coverage. We iterate over this process 100 times. 

Figure \ref{fig:lasso:example} and Table \ref{table:lasso:example} present our results. From the empirical CDF plots of $p$-values we see that the selective $p$-values are uniform when we use our sampler while far from uniform when we use the projected Langevin sampler for sampling the optimization variables.
The coverages of the constructed confidence intervals using the proposed sampler are preserved at the target level of $90\%$, while naive confidence intervals and the selective ones using projected Langevin MCMC are significantly under-covering the zero vector $\beta_E^*$. We run the projected Langevin for more number of steps taken in the respective chains to make the time between the two comparable. The results imply the Markov chain via IRF converges faster to the target distribution than the projected Langevin MC.

\begin{figure}[h] 
\centering 
\includegraphics[width=0.9\textwidth]{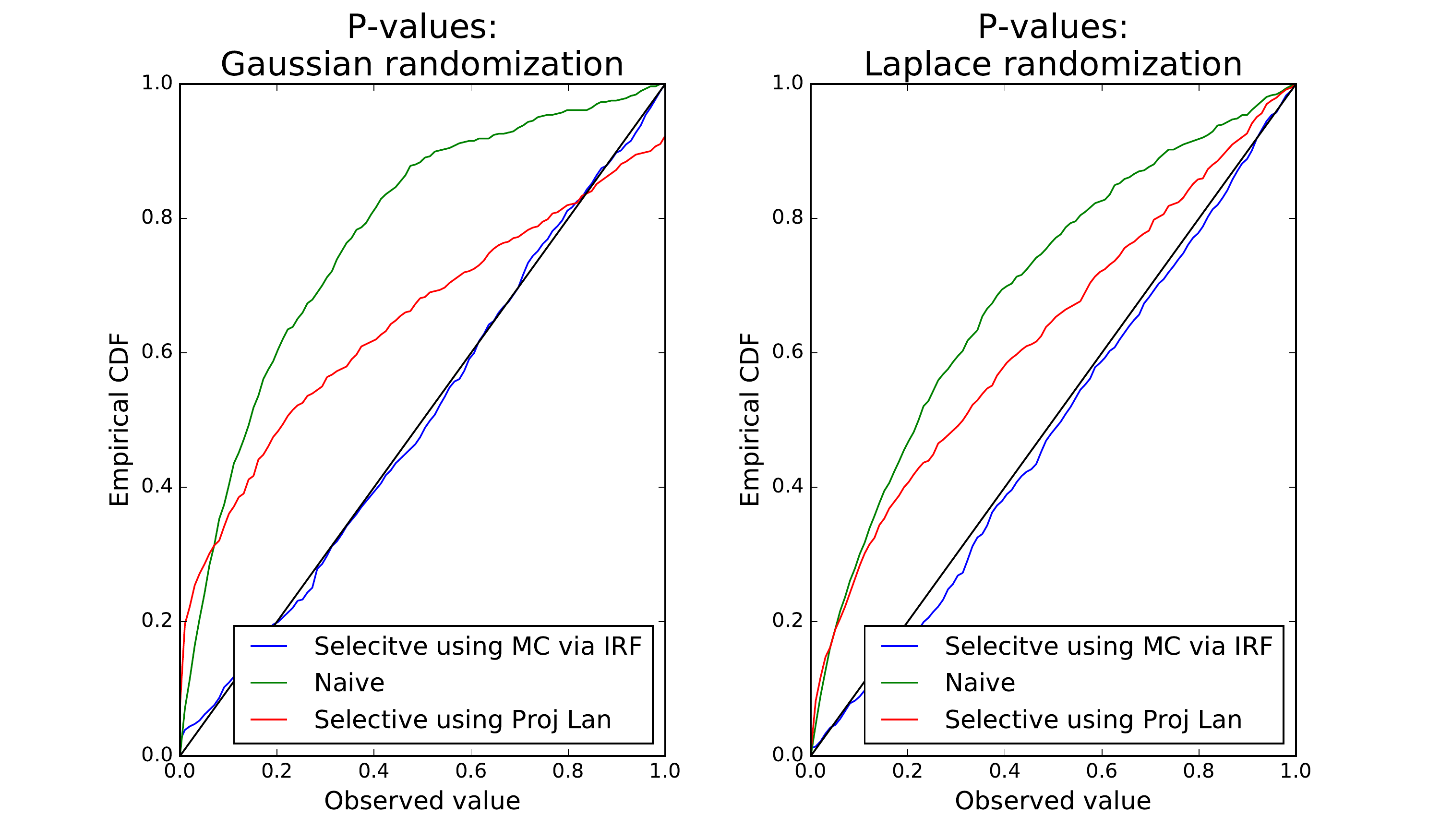}
\caption{Empirical CDF of selective and naive $p$-values corresponding to the selected set $E$ chosen based on randomized LASSO.}
\label{fig:lasso:example}
\end{figure}

\begin{table}[h!] 
\centering
 \begin{tabular}{||c| c c c |c c c c||} 
 \hline
     	& \multicolumn{3}{c}{\textit{Gaussian randomization}} & \multicolumn{3}{c}{\textit{Laplace randomization}} & \\
     	& coverage & $\#$ steps & time & coverage & $\#$ steps & time & \\ [0.5ex] 
 \hline \hline
 Markov chain via IRF & 89.59 & 1,000 & 73 &  91.39 & 1,000 & 143 & \\
 Projected Langevin   & 56.22 & 3,000 & 86 &  71.76 & 4,000 & 176 & \\
 Naive intervals	  & 59.97 & na & na & 70.31 & na & na & \\ 
 \hline
 \end{tabular}
 \caption{Confidence intervals coverages (in $\%$), where the targeted coverage is $90\%$. We also include the number of steps taken in each Markov chain and the total time (in seconds) of the program.}
 \label{table:lasso:example}
\end{table}

\section{Conclusion}

This work introduces a new discrete-time MCMC sampler, proving the target distribution is invariant for the corresponding kernel. At every step, it requires computing the first arrival time of the appropriate Poisson process. Computational simplifications are provided in specific univariate examples including some multimodal distributions. In high-dimensions, the method offers enormous computational improvement when used in selective inference. There remains the question of how fast the sampler mixes in the general case.

\section*{Acknowledgements} The authors would like to thank Persi Diaconis for useful comments and pointers to the literature on iterated random functions, and Jonathan Taylor for helpful feedback on application to selective inference.
\begin{appendices}

\section{Examples} \label{app:examples}

\subsection{Beta distribution \boldmath$B(\alpha,\beta)$}

\textbf{Example \ref{example:beta} continued.}

We present the computation of $\tau$ for other $\alpha$ and $\beta$ values.

\begin{enumerate}[leftmargin=*]
\setcounter{enumi}{1}	
\item $\alpha<1, \beta<1$. In this case $U$ is a concave function.
	\begin{itemize}[leftmargin=*]
	\item Case $v_n=1$. Denote $t^*=\textnormal{arg}\:\underset{t\geq 0}{\max}\: U(x_n+t)=\left(\frac{1}{1+\frac{\beta-1}{\alpha-1}}-x_n\right)_+$. Then $U'(x_n+t)>0$ for $t<t^*$ and $U'(x_n+t)<0$ for $t>t^*$.  To compute $\tau$, we differentiate between the following two cases with respect to $t^*$.
	\begin{enumerate}[label=(\alph*)]
	\item For $U(x_n+t^*)-U(x_n)\geq -\log V_n$, $\tau$ solves
	\begin{equation*}
		-\log V_n=\int_0^{\tau}U'(x_n+t)_+dt=\int_0^{\tau}U'(x_n+t)dt=U(x_n+\tau)-U(x_n),
	\end{equation*}
	$\tau\in[0,t^*]$. Solving this involves a numerical line search.
	\item For $U(x_n+t^*)-U(x_n)<-\log V_n$, $\tau=1-x_n$.
	\end{enumerate}
	
	\item Case $v_n=-1$. Denote $t^*=\textnormal{arg}\:\underset{t\geq 0}{\max}\: U(x_n-t)$. Then $U'(x_n-t)<0$ for $t<t^*$ and $U'(x_n-t)>0$ for $t>t^*$. To compute $\tau$, we differentiate between following two cases with respect to $t^*$.
	\begin{enumerate}[label=(\alph*)]
	\item For $U(x_n-t^*)-U(x_n)\geq -\log V_n$, $\tau$ solves
	\begin{equation*}
		-\log V_n=\int_0^{\tau}(-U'(x_n-t))_+dt=\int_0^{\tau}-U'(x_n-t)dt=U(x_n-\tau)-U(x_n),
	\end{equation*}
	$\tau\in[0,t^*]$. 
	\item For $U(x_n-t^*)-U(x_n)<-\log V_n$, $\tau=1-x_n$.
	\end{enumerate}
	\end{itemize}

\item $\alpha<1, \beta>1$. In this case $U'(x)> 0$ for all $x\in(0,1)$, hence $U(x)$ increasing on $(0,1)$.
	\begin{itemize}[leftmargin=*]
	\item Case $v_n=1$. $\tau$ solves
	\begin{equation*}
		-\log V_n=\int_0^{\tau}U'(x_n+t)_+dt=\int_0^{\tau}U'(x_n+t)dt=U(x_n+\tau)-U(x_n),
	\end{equation*}
	$\tau\in(0,1-x_n)$.
	\item Case $v_n=-1$. $\tau=x_n$
	\end{itemize}
	
\item $\alpha>1, \beta<1$. In this case $U'(x)<0$ for all $x\in(0,1)$, hence $U(x)$ decreasing on $(0,1)$.
	\begin{itemize}[leftmargin=*]
	\item Case $v_n=1$. $\tau=1-x_n$.
	\item Case $v_n=-1$. $\tau$ solves
	\begin{equation*}
	\begin{aligned}
		-\log V_n &=\int_0^{\tau}(-U'(x_n-t))_+dt=\int_0^{\tau}-U'(x_n-t)dt=	 \int_0^{-\tau}U'(x_n+s)ds \\
		&=U(x_n-\tau)-U(x_n),
	\end{aligned}
	\end{equation*}
	$\tau\in(0,x_n)$.
	\end{itemize}
\end{enumerate}


\textbf{Example \ref{example:beta:decomposed} continued.}

We present the updates for another MCMC using the decomposition of $U$ into a sum of increasing and decreasing function, $U_1$ and $U_2$ respectively.
\begin{enumerate}[leftmargin=*]
\setcounter{enumi}{1}
		\item $\alpha<1, \beta<1$: $U(x)=U_1(x)+U_2(x)$, $x\in(0,1)$, for an increasing function $U_1(x)=-(\alpha-1)\log x$ and a decreasing function $U_2(x)=-(\beta-1)\log (1-x)$. To simulate $\tau$, we differentiate between the following two cases of positive and negative velocity.
	\begin{itemize}[leftmargin=*]
	\item Case $v_n=1$: $\tau=\textnormal{min}\{U_1^{-1}(U_1(x_n)-\log V_n)-x_n, 1-x_n\}$.
 	\item Case $v_n=-1$: $\tau=\textnormal{min}\{x_n-U_2^{-1}(U_2(x_n)-\log V_n),x_n\}$.
	\end{itemize}

	\item $\alpha>1, \beta<1$: $U(x)=U_2(x)$, $x\in(0,1)$, for a decreasing function $U_2(x)=-(\alpha-1)\log x-(\beta-1)\log (1-x)$.	
	\begin{itemize}[leftmargin=*]
	\item Case $v_n=1$: $\tau=1-x_n$.
	\item Case $v_n=-1$: $\tau=\textnormal{min}\{x_n-U_2^{-1}(U_2(x_n)-\log V_n),1-x_n\}$, $U_2^{-1}(\cdot)$ computed numerically.
	\end{itemize}

	\item $\alpha<1, \beta>1$: $U(x)=U_1(x)$, $x\in(0,1)$, for an increasing function $U_1(x)=-(\alpha-1)\log x-(\beta-1)\log (1-x)$.
	\begin{itemize}[leftmargin=*]
	\item Case $v_n=1$: $\tau=\textnormal{min}\{U_1^{-1}(U_1(x_n)-\log V_n)-x_n, 1-x_n\}$, $U_1^{-1}(\cdot)$ computed numerically.
	\item Case $v_n=-1$: $\tau=x_n$.
	\end{itemize}
\end{enumerate}

\subsection{Gaussian distribution $\mathcal{N}(\mu,\sigma^2)$}

We illustrate the sampling in the case of Gaussian distribution $\mathcal{N}(\mu,\sigma^2)$. Using the sampler from Section \ref{sec:sampler}, the updates are as follows. Given the current position $x_n$ and the sampled parameters $V_n$ and $v_n$, we have $\lambda(t)=(U'(x_n+v_nt)v_n)_+$, where $U(x)=\frac{(x-\mu)^2}{2\sigma^2}$ is negative log-density of the Gaussian $\mathcal{N}(\mu,\sigma^2)$ distribution. Denote $t^*=\textnormal{arg}\:\underset{t\geq 0}{\min}\:U(x_n+v_nt)=\left(\frac{\mu-x_n}{v_n}\right)_+$. To compute $\tau$, we differentiate between the following two cases.
\begin{itemize}[leftmargin=*]
\item Case $\frac{\mu-x_n}{v_n}\geq 0$. $\tau\geq t^*=\frac{\mu-x_n}{v_n}$ solves $U(x_n+v_n\tau)=-\log V_n$, giving us $\tau$ explicitly
\begin{equation*}
		\tau=\frac{\mu-x_n}{v_n}+\frac{1}{|v_n|}\sqrt{2\sigma^2(-\log V_n)}.
\end{equation*}
\item Case $\frac{\mu-x_n}{v_n}<0$. $\tau\geq t^*=0$ solves $U(x_n+v_n\tau)-U(x_n)=-\log V_n$, thus
\begin{equation*}
	\tau=\frac{\mu-x_n}{v_n}+\sqrt{\frac{(x_n-\mu)^2}{v_n^2}+\frac{2\sigma^2(-\log V_n)}{v_n^2}}.	
\end{equation*}

As described in Section \ref{sec:sampling:via:decomposition}, another sampler uses the decomposition $U=U_1+U_2$, with $U_1(x)=\frac{(x-\mu)^2}{2\sigma^2}\mathbb{I}_{\{x\geq \mu\}}$ and $U_2(x)=\frac{(x-\mu)^2}{2\sigma^2}\mathbb{I}_{\{x\leq \mu\}}$. Given $U_1$ and $U_2$ this sampler has an explicit expression for $\tau$.
\end{itemize}

\section{Proofs} \label{app:proofs}

\begin{proof-of-theorem}{\ref{thm:stationarity:decomposition}}
Since
\begin{equation}
\begin{aligned}
	\int_{-\infty}^{\infty}\pi(x_n)K(x_n,x_{n+1})dx_n
	&=\int_{-\infty}^{x_{n+1}}U_1'(2x_{n+1}-x_n)e^{-U_1(2x_{n+1}-x_n)+U_1(x_n)}\pi(x_n)dx_n \\
	&+\int_{x_{n+1}}^{\infty}-U_2'(2x_{n+1}-x_n)e^{-U_2(2x_{n+1}-x_n)+U_2(x_n)}\pi(x_n)dx_n,
\end{aligned}
\end{equation}
differentiating the above with respect to $x_{n+1}$ we get 
\begin{equation*}
\begin{aligned}
	&U_1'(x_{n+1})\pi(x_{n+1}) +\int_{-\infty}^{x_{n+1}}\frac{d}{dx_{n+1}}\left(U_1'(2x_{n+1}-x_n)e^{-U_1(2x_{n+1}-x_n)+U_1(x_n)}\right)\pi(x_n)dx_n\\
	&+U_2'(x_{n+1})\pi(x_{n+1})+\int_{x_{n+1}}^{\infty}\frac{d}{dx_{n+1}}\left(-U_2'(2x_{n+1}-x_n)e^{-U_2(2x_{n+1}-x_n)+U_2(x_n)}\right)\pi(x_n)dx_n.
\end{aligned}
\end{equation*}
To prove $\pi(x)\:\propto\:e^{-U(x)}$ is stationary for this kernel, it suffices to show
\begin{equation*} 
\begin{aligned}
	2\pi'(x_{n+1})&=-2U'(x_{n+1})\pi(x_{n+1}) \\
	&=\int_{-\infty}^{x_{n+1}}\frac{d}{dx_{n+1}}\left(U_1'(2x_{n+1}-x_n)e^{-U_1(2x_{n+1}-x_n)}\right)e^{U_1(x_n)}\pi(x_n)dx_n \\
	&\qquad +\int_{x_{n+1}}^{\infty}\frac{d}{dx_{n+1}}\left(-U_2'(2x_{n+1}-x_n)e^{-U_2(2x_{n+1}-x_n)}\right)e^{U_2(x_n)}\pi(x_n)dx_n,
\end{aligned}
\end{equation*}
i.e.~suffices to show
\begin{equation}\label{eq:suff:stationarity:conditon}
\begin{aligned}
	&-2U'(x_{n+1})e^{-U(x_{n+1})}\\
	&=\int_{-\infty}^{x_{n+1}}\left(2U_1^{''}(2x_{n+1}-x_n)-2{U_1'}^2(2x_{n+1}-x_n)\right)e^{-U_1(2x_{n+1}-x_n)}e^{-U_2(x_n)}dx_n \\
	&\qquad +\int_{x_{n+1}}^{\infty}\left(-2U_2^{''}(2x_{n+1}-x_n)+2{U_2'}^2(2x_{n+1}-x_n)\right)e^{-U_2(2x_{n+1}-x_n)}e^{-U_1(x_n)}dx_n.
\end{aligned}
\end{equation}
After the change of variables $u=x_n$ in the first integral and $u=2x_{n+1}-x_n$ in the second integral, the right hand side of \eqref{eq:suff:stationarity:conditon} becomes
\begin{equation*}
\begin{aligned}
	&2\int_{-\infty}^{x_{n+1}}\left(U_1^{''}(2x_{n+1}-u)-{U_1'}^2(2x_{n+1}-u)\right)e^{-U_1(2x_{n+1}-u)}e^{-U_2(u)}du \\
	&\qquad+2\int_{-\infty}^{x_{n+1}}\left(-U_2^{''}(u)+{U_2'}^2(u)\right)e^{-U_2(u)}e^{-U_1(2x_{n+1}-u)}du \\
	&=2\int_{-\infty}^{x_{n+1}}\left(U_1^{''}(2x_{n+1}-u)-U_2^{''}(u)+{U_2'}^2(u)-{U_1'}^2(2x_{n+1}-u)\right)e^{-U_1(2x_{n+1}-u)}e^{-U_2(u)}du\\
	&=2\int_{-\infty}^{x_{n+1}}\left(U_1^{''}(2x_{n+1}-u)-U_2^{''}(u)\right)e^{-U_1(2x_{n+1}-u)-U_2(u)}du \\
	&\qquad +2\int_{-\infty}^{x_{n+1}}\left(U_1'(2x_{n+1}-u)-U_2'(u)\right)\left(-U_1'(2x_{n+1}-u)-U_2'(u)\right)e^{-U_1(2x_{n+1}-u)-U_2(u)}du. 
\end{aligned}	
\end{equation*}
Denoting $\kappa(u)=-U_1'(2x_{n+1}-u)-U_2'(u)$ and $g(u)=e^{-U_1(2x_{n+1}-u)-U_2(u)}$, the above expression becomes
\begin{equation*}
\begin{aligned}
	&2\int_{-\infty}^{x_{n+1}}\kappa'(u)g(u)du+2\int_{-\infty}^{x_{n+1}}\kappa(u)g'(u)du=2\int_{-\infty}^{x_{n+1}}(\kappa(u)g(u))'du \\
	&=2\kappa(x_{n+1})g(x_{n+1})
	=-2(U_1'(x_{n+1})+U_2'(x_{n+1}))e^{-U(x_{n+1})}=-2U'(x_{n+1})e^{-U(x_{n+1})}.
\end{aligned}	
\end{equation*}

\end{proof-of-theorem}

\end{appendices}

\bibliography{cite}

\begin{thebibliography}{27}
\providecommand{\natexlab}[1]{#1}
\providecommand{\url}[1]{\texttt{#1}}
\expandafter\ifx\csname urlstyle\endcsname\relax
  \providecommand{\doi}[1]{doi: #1}\else
  \providecommand{\doi}{doi: \begingroup \urlstyle{rm}\Url}\fi

\bibitem[Andersen and Diaconis(2007)]{andersen2007hit}
Hans~C Andersen and Persi Diaconis.
\newblock Hit and run as a unifying device.
\newblock \emph{Journal de la soci{\'e}t{\'e} fran{\c{c}}aise de statistique},
  148\penalty0 (5):\penalty0 5--28, 2007.

\bibitem[Bi et~al.(2017)Bi, Markovic, Xia, and Taylor]{bi2017inferactive}
Nan Bi, Jelena Markovic, Lucy Xia, and Jonathan Taylor.
\newblock Inferactive data analysis.
\newblock \emph{arXiv preprint arXiv:1707.06692}, 2017.

\bibitem[Bierkens(2016)]{bierkens2016non}
Joris Bierkens.
\newblock Non-reversible {M}etropolis-{H}astings.
\newblock \emph{Statistics and Computing}, 26\penalty0 (6):\penalty0
  1213--1228, 2016.

\bibitem[Bierkens et~al.(2016)Bierkens, Fearnhead, and
  Roberts]{bierkens2016zig}
Joris Bierkens, Paul Fearnhead, and Gareth Roberts.
\newblock The zig-zag process and super-efficient sampling for bayesian
  analysis of big data.
\newblock \emph{arXiv preprint arXiv:1607.03188}, 2016.

\bibitem[Bierkens et~al.(2017)Bierkens, Bouchard-C{\^o}t{\'e}, Doucet, Duncan,
  Fearnhead, Roberts, and Vollmer]{bierkens2017piecewise}
Joris Bierkens, Alexandre Bouchard-C{\^o}t{\'e}, Arnaud Doucet, Andrew~B
  Duncan, Paul Fearnhead, Gareth Roberts, and Sebastian~J Vollmer.
\newblock Piecewise deterministic markov processes for scalable monte carlo on
  restricted domains.
\newblock \emph{arXiv preprint arXiv:1701.04244}, 2017.

\bibitem[Bouchard-C{\^o}t{\'e} et~al.(2017)Bouchard-C{\^o}t{\'e}, Vollmer, and
  Doucet]{bouchard2017bouncy}
Alexandre Bouchard-C{\^o}t{\'e}, Sebastian~J Vollmer, and Arnaud Doucet.
\newblock The bouncy particle sampler: A non-reversible rejection-free {M}arkov
  chain {M}onte {C}arlo method.
\newblock \emph{Journal of the American Statistical Association}, \penalty0
  (just-accepted), 2017.

\bibitem[Bubeck et~al.(2015)Bubeck, Eldan, and Lehec]{bubeck2015sampling}
S{\'e}bastien Bubeck, Ronen Eldan, and Joseph Lehec.
\newblock Sampling from a log-concave distribution with projected langevin
  monte carlo.
\newblock \emph{arXiv preprint arXiv:1507.02564}, 2015.

\bibitem[Chen and Hwang(2013)]{chen2013accelerating}
Ting-Li Chen and Chii-Ruey Hwang.
\newblock Accelerating reversible {M}arkov chains.
\newblock \emph{Statistics \& Probability Letters}, 83\penalty0 (9):\penalty0
  1956--1962, 2013.

\bibitem[Diaconis and Freedman(1999)]{diaconis1999iterated}
Persi Diaconis and David Freedman.
\newblock Iterated random functions.
\newblock \emph{SIAM review}, 41\penalty0 (1):\penalty0 45--76, 1999.

\bibitem[Diaconis et~al.(2000)Diaconis, Holmes, and Neal]{diaconis2000analysis}
Persi Diaconis, Susan Holmes, and Radford~M Neal.
\newblock Analysis of a nonreversible markov chain sampler.
\newblock \emph{Annals of Applied Probability}, pages 726--752, 2000.

\bibitem[Duane et~al.(1987)Duane, Kennedy, Pendleton, and
  Roweth]{duane1987hybrid}
Simon Duane, Anthony~D Kennedy, Brian~J Pendleton, and Duncan Roweth.
\newblock Hybrid monte carlo.
\newblock \emph{Physics letters B}, 195\penalty0 (2):\penalty0 216--222, 1987.

\bibitem[Hwang et~al.(1993)Hwang, Hwang-Ma, and Sheu]{hwang1993accelerating}
Chii-Ruey Hwang, Shu-Yin Hwang-Ma, and Shuenn-Jyi Sheu.
\newblock Accelerating gaussian diffusions.
\newblock \emph{The Annals of Applied Probability}, pages 897--913, 1993.

\bibitem[Lee and Taylor(2014)]{lee_screening}
Jason~D Lee and Jonathan~E Taylor.
\newblock Exact post model selection inference for marginal screening.
\newblock In \emph{Advances in Neural Information Processing Systems}, pages
  136--144, 2014.

\bibitem[Lee et~al.(2016)Lee, Sun, Sun, and Taylor]{lee2013exact}
Jason~D Lee, Dennis~L Sun, Yuekai Sun, and Jonathan~E Taylor.
\newblock Exact post-selection inference, with application to the lasso.
\newblock \emph{The Annals of Statistics}, 44\penalty0 (3):\penalty0 907--927,
  2016.

\bibitem[Leli{\`e}vre et~al.(2013)Leli{\`e}vre, Nier, and
  Pavliotis]{lelievre2013optimal}
Tony Leli{\`e}vre, Francis Nier, and Grigorios~A Pavliotis.
\newblock Optimal non-reversible linear drift for the convergence to
  equilibrium of a diffusion.
\newblock \emph{Journal of Statistical Physics}, 152\penalty0 (2):\penalty0
  237--274, 2013.

\bibitem[Markovic and Taylor(2016)]{bootstrap_mv}
Jelena Markovic and Jonathan Taylor.
\newblock Bootstrap inference after using multiple queries for model selection.
\newblock \emph{arXiv preprint arXiv:1612.07811}, 2016.

\bibitem[Neal et~al.(2011)]{neal2011mcmc}
Radford~M Neal et~al.
\newblock Mcmc using hamiltonian dynamics.
\newblock \emph{Handbook of Markov Chain Monte Carlo}, 2\penalty0 (11), 2011.

\bibitem[Pakman(2017)]{pakman2017binary}
Ari Pakman.
\newblock Binary bouncy particle sampler.
\newblock \emph{arXiv preprint arXiv:1711.00922}, 2017.

\bibitem[Pakman et~al.(2016)Pakman, Gilboa, Carlson, and
  Paninski]{pakman2016stochastic}
Ari Pakman, Dar Gilboa, David Carlson, and Liam Paninski.
\newblock Stochastic bouncy particle sampler.
\newblock \emph{arXiv preprint arXiv:1609.00770}, 2016.

\bibitem[Peters and de~With(2012)]{peters2012rejection}
Elias~AJF Peters and G.~de~With.
\newblock Rejection-free {M}onte {C}arlo sampling for general potentials.
\newblock \emph{Physical Review E}, 85\penalty0 (2):\penalty0 026703, 2012.

\bibitem[Sun et~al.(2010)Sun, Schmidhuber, and Gomez]{sun2010improving}
Yi~Sun, J{\"u}rgen Schmidhuber, and Faustino~J Gomez.
\newblock Improving the asymptotic performance of markov chain monte-carlo by
  inserting vortices.
\newblock In \emph{Advances in Neural Information Processing Systems}, pages
  2235--2243, 2010.

\bibitem[Taylor and Tibshirani(2015)]{taylor2015statistical}
Jonathan Taylor and Robert~J Tibshirani.
\newblock Statistical learning and selective inference.
\newblock \emph{Proceedings of the National Academy of Sciences}, 112\penalty0
  (25):\penalty0 7629--7634, 2015.

\bibitem[Tian and Taylor(2015)]{tian2015selective}
Xiaoying Tian and Jonathan~E Taylor.
\newblock Selective inference with a randomized response.
\newblock \emph{Annals of Statistics to apper}, 2015.

\bibitem[Tian et~al.(2016)Tian, Panigrahi, Markovic, Bi, and
  Taylor]{selective_sampler}
Xiaoying Tian, Snigdha Panigrahi, Jelena Markovic, Nan Bi, and Jonathan Taylor.
\newblock Selective sampling after solving a convex problem.
\newblock \emph{arXiv preprint arXiv:1609.05609}, 2016.

\bibitem[Tibshirani et~al.(2016)Tibshirani, Taylor, Lockhart, and
  Tibshirani]{sequential_post_selection}
Ryan~J. Tibshirani, Jonathan Taylor, Richard Lockhart, and Robert Tibshirani.
\newblock Exact post-selection inference for sequential regression procedures.
\newblock \emph{Journal of the American Statistical Association}, 111\penalty0
  (514):\penalty0 600--620, 2016.

\bibitem[Turitsyn et~al.(2011)Turitsyn, Chertkov, and
  Vucelja]{turitsyn2011irreversible}
Konstantin~S Turitsyn, Michael Chertkov, and Marija Vucelja.
\newblock Irreversible monte carlo algorithms for efficient sampling.
\newblock \emph{Physica D: Nonlinear Phenomena}, 240\penalty0 (4):\penalty0
  410--414, 2011.

\bibitem[Wu and Robert(2017)]{wu2017generalized}
Changye Wu and Christian~P Robert.
\newblock Generalized bouncy particle sampler.
\newblock \emph{arXiv preprint arXiv:1706.04781}, 2017.

\end{thebibliography}
\bibliographystyle{plainnat}

	
\end{document}